\documentclass[runningheads,USenglish]{llncs}
\usepackage[T1]{fontenc}
\usepackage{graphicx}
\graphicspath{{images/}}
\usepackage[disable]{todonotes}
\usepackage{microtype}
\usepackage[font=small, labelfont=bf]{caption}
\usepackage[font=small, labelfont=small]{subcaption}
\usepackage{xspace,amssymb}\usepackage{amssymb}
\usepackage{xcolor}
\usepackage{amsmath}
\definecolor{defblue}{rgb}{0.121,0.47,0.705}
\definecolor{linkblue}{rgb}{0.098,0.098,0.4392}
\let\emph\relax
\DeclareTextFontCommand{\emph}{\color{defblue}\em}
\DeclareTextFontCommand{\bl}{\color{defblue}}
\usepackage[colorlinks=true, 
    linkcolor=linkblue, 
    anchorcolor=linkblue,
    citecolor=linkblue, 
    filecolor=linkblue,
    menucolor=linkblue,
    urlcolor=linkblue, 
    bookmarks=true,
    bookmarksopen=true,
    bookmarksopenlevel=2, 
    bookmarksnumbered=true, 
    hyperindex=true,
    plainpages=false, 
    pdfpagelabels=true
]{hyperref}
\usepackage[capitalise,noabbrev,nameinlink]{cleveref}
\usepackage[inline]{enumitem}
\renewcommand{\orcidID}[1]{\href{https://orcid.org/#1}{\includegraphics[scale=.03]{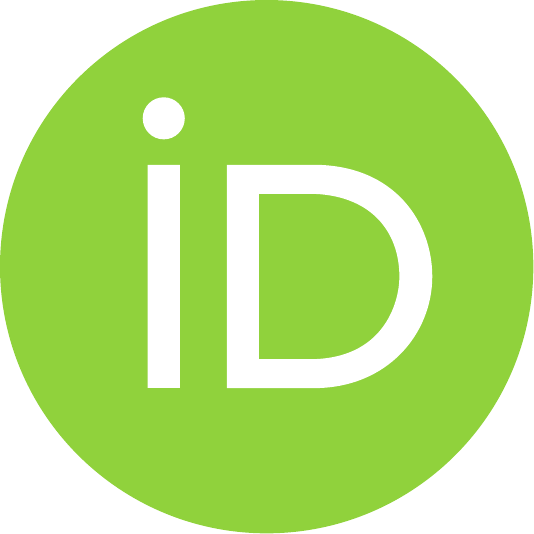}}}

\usepackage{multibib}
\newcites{app}{References (Appendix only)}

\usepackage{listings} 			%
\usepackage{rotating}

\usepackage[ruled,vlined]{algorithm2e}
\SetKw{KwAnd}{and}
\SetKw{KwOr}{or}
\makeatletter
\def\BState{\State\hskip-\ALG@thistlm}
\makeatother
\usepackage{thmtools,apptools,thm-restate}
\newcommand{\restateref}[1]{\IfAppendix{\hyperref[#1]{$\star$}}{\hyperref[#1*]{$\star$}}}

\let\doendproof\endproof
\renewcommand\endproof{~\hfill$\qed$\doendproof}
\spnewtheorem*{proofWithMath}{Proof}{\itshape}{\rmfamily}

\usepackage{etoolbox}
\BeforeBeginEnvironment{algorithm}{\definecolor{defblue}{rgb}{0,0,0}}
\AfterEndEnvironment{algorithm}{\definecolor{defblue}{rgb}{0.121,0.47,0.705}}

\newcommand{\lc}{\textsc{Layered\-Crown}\xspace}
\newcommand{\mlc}{\textsc{Max-\lc}\xspace}
\newcommand{\mmlc}{\textsc{(Max-)\lc}\xspace}
\newcommand{\ilc}{\textsc{Int\-\lc}\xspace}
\newcommand{\milc}{\textsc{Max-\ilc}\xspace}

\newcommand{\true}{\texttt{true}\xspace}
\newcommand{\false}{\texttt{false}\xspace}
\newcommand{\Oh}{\mathcal{O}}

\DeclareMathOperator{\ALG}{ALG}
\DeclareMathOperator{\OPT}{OPT}
\DeclareMathOperator{\poly}{poly}

\begin{document}
\title{On Layered Area-Proportional Rectangle Contact Representations}

\author{Carolina Haase\inst{1}\orcidID{0000-0001-6696-074X} \and
Philipp Kindermann\inst{1}\orcidID{0000-0001-5764-7719}}

\institute{Universit\"at Trier, Trier, Germany\\
\email{\{haasec,kindermann\}@uni-trier.de}}

\maketitle              %
\begin{abstract}
Semantic word clouds visualize the semantic relatedness between the words of a text by placing pairs of related words close to each other. Formally, the problem of drawing semantic word clouds corresponds to drawing a rectangle contact representation of a graph whose vertices correlate to the words to be displayed and whose edges indicate that two words are semantically related.
The goal is to maximize the number of realized contacts while avoiding any false adjacencies.
We consider a variant of this problem that restricts input graphs to be layered and all rectangles to be of equal height, called \textsc{Maximum Layered Contact Representation Of Word Networks} or \mlc, as well as the variant \milc, which restricts the problem to only rectangles of integer width and the placement of those rectangles to integer coordinates.

We classify the corresponding decision problem $k$-\ilc as NP-complete even for triangulated graphs and $k$-\lc as NP-complete for planar graphs. We introduce three algorithms: a 1/2-approximation for \mlc of triangulated graphs, and a PTAS and an XP algorithm for \milc with rectangle width polynomial in $n$.
\end{abstract}
\section{Introduction}
Word clouds can be used to visualize the importance of (key-)words in a given text. Usually, words will be scaled according to their frequency and, in case of semantic word clouds, arranged in such a way that closely related words are placed closer together than words that are unrelated. 
There are multiple tools like Wordle\footnote{At the time of writing, the tool (usually found at \url{http://www.wordle.net/}) is not available, but the creator states on their website (\url{https://mrfeinberg.com/}) that they have ``hopes to bring it back to life''.} \cite{Viegas2009}, which was launched in 2008 by Jonathan Feinberg, 
that allow for automized drawing of classical word clouds, i.e., word clouds that disregard semantic relatedness; see \cref{img:wordclouds} for an example.

However, classical word clouds have certain disadvantages, as they are frequently misinterpreted. This has been analyzed in a survey conducted by Viegas et al.~\cite{Viegas2009}: different colors and positioning of words give the impression to bear meaning, even if they don't.
For this reason, it makes sense to pay special attention to \emph{semantic} word clouds, which resolve these 
shortcomings by placing related words closely together and sometimes using color to indicate, for example,  
clusters of semantically related words. 
Semantic relatedness, in this case, can be measured by how often two words occur together in the same sentence~\cite{Barth2014}.

Tools to generate semantic word clouds are, however, not as widely available. One such tool can be found online at \url{http://wordcloud.cs.arizona.edu} that implements different algorithms for semantic word clouds~\cite{Barth:14,Barth2014,Bekos2014}. A semantic word cloud generated by the tool 
is shown in \cref{img:wordclouds}. 
In the given example, the placement of words was calculated using cosine similarity. 
Compared to the classical word cloud generated using the same tool, with the same coloring for clusters, but a greedy, randomized approach to place words, the advantages of arranging words semantically become quite clear.

\begin{figure}[!t]
	\centering
	\includegraphics[scale=0.2]{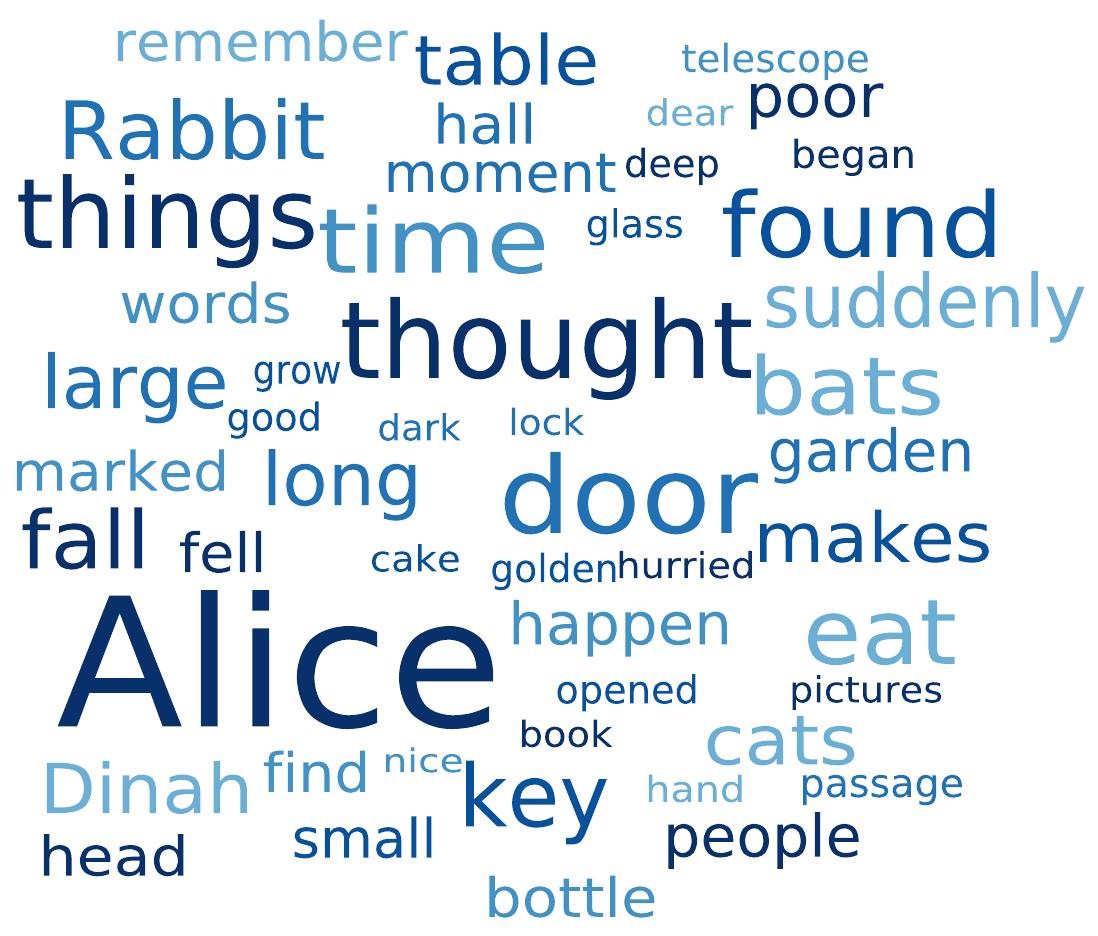}
	\includegraphics[scale=0.3]{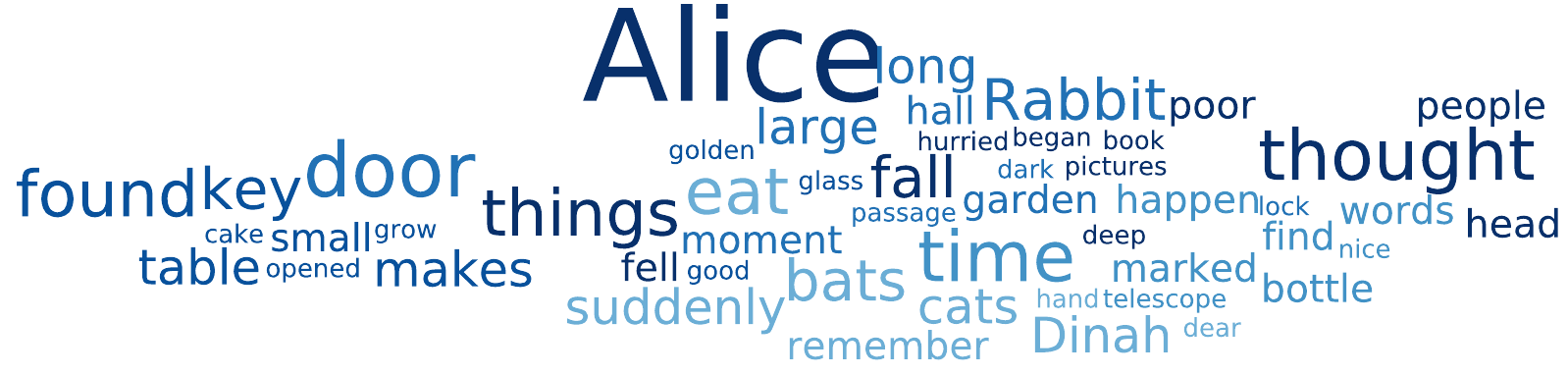}
	\caption{Randomly arranged word cloud (left) and semantic word cloud (right), generated using the first chapter of ``Alice's Adventures in Wonderland'' by Lewis Carroll.}
	\label{img:wordclouds}
\end{figure}

\paragraph{Problem statement.}
To formalize the problem of drawing semantic word clouds, Barth et al.~\cite{Barth:14} introduced the problem \bl{\textsc{Contact Representation Of Word Networks} (\textsc{Crown})}.
Given a graph $G=(V,E)$, where every vertex $v_i$ of $G$ corresponds to a word of width $w_i$ and height $h_i$, and 
every (weighted) edge between two vertices indicates the level of semantic relatedness between the corresponding 
words, the goal is to draw a contact representation where each vertex $v_i$ is drawn as an axis-aligned rectangle
of width $w_i$ and height $h_i$ such that bounding boxes of semantically related words touch.

In this paper, we consider a more restricted variant of the problem, which we will call \bl{\mmlc}, that has been introduced by Nöllenburg et al.~\cite{Noellenburg:21}.
Here, the input graph is planar and the vertices are assigned to layers. Furthermore, all 
bounding boxes have the same height. The goal is to maximize the number of contacts between semantically related 
words, while words that are not semantically related are not allowed to touch.

More formally, the problem is defined as follows. Let $G = (V,E)$ be a planar vertex-weighted \emph{layered graph} with $L$ 
layers, i.e., each vertex is assigned to one of $L$ layers.
The order of vertices within a layer is fixed, i.e., each vertex
$v_{i,j}$ can be identified by its layer $1 \leq i \leq L$ and its position $j$ within the layer. Edges can only exist between
neighboring vertices $v_{i,j}, v_{i,j+1},$ on the same layer and between vertices on adjacent layers. 
Like Nöllenburg et al., we consider the case that the edges are unweighted.
To each vertex $v$ we assign an axis-aligned unit-height rectangle $R(v)$ with width $w(v)$, given by the weight of 
the vertex. 
We will also use the notation $R_{i,j}=R(v_{i,j})$ and $w_{i,j}=w(v_{i,j})$; see \cref{img:layered-crown}.
The goal is to calculate the position $x_{i,j}$ for each vertex $v_{i,j}$, where $x_{i,j}$ denotes the $x$-coordinate 
of the bottom left corner of $R_{i,j}$, in such a way that rectangles do not overlap except on their boundaries. We 
call such an assignment a \emph{representation}.
Two rectangles $R(v)$ and $R(u)$ touch if their intersection is a line segment of length 
$\varepsilon > 0$. In this case, we say that $R(v)$ and $R(u)$ are in \emph{contact}.
An edge $\{v,u\}$ is \emph{realized} if $R_v$ and $R_u$ are in contact. We call a contact \emph{horizontal} if $R_v$ 
and $R_u$ are neighbors on the same layer and \emph{vertical} if $R_v$ and $R_u$ are on adjacent layers.
Contacts between rectangles whose vertices are not adjacent are not allowed and are called \emph{false adjacencies}. 
Representations with false adjacencies are \emph{invalid}; otherwise, they are \emph{valid}. 
Gaps between vertices 
$v_{i,j}, v_{i,j+1}$ on the same layer are allowed.

The maximization problem \bl{\textsc{Maximum Layered Contact Representation of Word Networks} (\mlc)} is to find a 
valid representation for a given graph $G$ such that the number of realized contacts is maximized.  The respective 
decision problem \bl{\textsc{Layered Contact Representation of Word Networks} ($k$-\lc)} is to decide whether there 
exists a valid contact representation that realizes at least $k$ contacts.
Many fonts are monospaced, i.e., all letters and characters occupy the same amount of horizontal space. Thus, we also consider the further restriction  that rectangles may only be of integer width and may only be placed with 
their lower left corner on integer coordinates. This implies that two rectangles are in contact if and only if the 
intersection of their boundaries is a line segment of positive integer length. 
We call those problems \bl{\milc} and \bl{$k$-\ilc}.

For information about graph drawing and parameterized complexity in general, we refer to books~\cite{DiBattista:99,Tamassia:13,Cygan2015,Downey:13%
} and \cref{ch:prelim}.

\begin{figure}[!t]
	\includegraphics{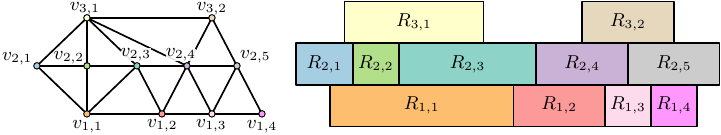}
	\caption{Internally triangulated graph with 3 layers (left) and a contact representation (right).}
	\label{img:layered-crown}
\end{figure}

\paragraph{Related work.}
Barth et al.~\cite{Barth:14} have shown that \textsc{Crown} is strongly NP-hard even when restricted to trees and weakly NP-hard even when restricted to stars, but can be solved in linear time on irreducible triangulations. They also provided constant-factor approximation algorithms for several graph classes like stars, trees, and planar graphs. These were improved by Bekos et al.~\cite{Bekos2014} and partially implemented and compared to other algorithms by Barth et al~\cite{Barth2014}.

Another variant of \textsc{Crown}, called \textsc{Hier-Crown}, restricts the input to be a directed acyclic graph with a single source and a plane embedding.
\textsc{Hier-Crown} can be solved in polynomial time, but can be shown to become weakly NP-complete if rectangles are allowed to be rotated  \cite{Barth:14}.

Barth et al.~\cite{Barth:14} further introduced another variant called \textsc{Area-Crown}, where the optimization goal shifts from maximizing rectangle contacts to minimizing the area of a bounding box containing the contact representation.
They show that this problem is NP-hard, even if restricted to paths.

Nöllenburg et al.~\cite{Noellenburg:21} introduced \mlc, but they only considered triangulated graphs. They gave a linear-time algorithm for triangulated graphs with only 2 layers and proposed an ILP-formulation for triangulated graphs with more than 2 layers. They further showed how to solve \textsc{Area-\lc} in polynomial time with a flow formulation.

Espenant and Mondal~\cite{Espenant:22} study \textsc{StreamTables}, where one seeks to visualize a matrix such that each cell is drawn as a rectangle of a specified area, cells in the same row have uniform height and align horizontally, while maximizing contacts and/or minimizing excess area. Their model is similar to \lc on grids, but false adjacencies are not forbidden, point contacts count as realized edges, and rows can generally be permuted. 

A larger overview on different kinds of word clouds and algorithms to solve them can be found in~\cref{ch:related}.

\paragraph{Our contribution.}
In this work, we study the computational complexity of \ilc and algorithms for \mlc and \milc.
In \cref{ch:complexity}, we classify $k$-\ilc as an NP-complete problem even for triangulated graphs, using a reduction from \textsc{Planar Monotone 3-Sat}. We will then adjust the proof to show NP-completeness for $k$-\lc for planar graphs.
In \cref{ch:algorithms}, we present a 1/2-approximation for \mlc on triangulated graphs (\cref{sec:2-approx}) and formulate a dynamic program for \milc that is an XP algorithm if the maximum rectangle width is polynomial in $n$ (\cref{sec:dp-alg}). Finally, we combine the ideas of the two algorithms to formulate a polynomial-time approximation scheme for \milc if the maximum rectangle width is polynomial in~$n$ (\cref{sec:ptas}).
We conclude with a list of research questions in Section \ref{ch:discussion}.
\todo{Hier sollten wir Mitte Seite 4 sein.}

\section{NP-completeness of $k$-\ilc}\label{ch:complexity}
In this section, we prove that $k$-\ilc is NP-complete. We first show that $k$-\ilc lies in NP.

\begin{lemma}\label{lem:lp-crown-np}
	$k$-\ilc lies in NP.
\end{lemma}
\begin{proof}
	For a given contact representation of a layered graph $G$, one can verify in polynomial time if the representation is valid and whether at least $k$ contacts are realized. Thereby, $k$-\ilc is a member of the class NP.
\end{proof}

We prove NP-hardness by reducing from \textsc{Planar Monotone 3-Sat}, which is NP-complete~\cite{DeBerg:12}.
Let $B$ be a boolean formula in conjunctive normal form (CNF) and $X = \{x_1,\dots,x_n\}$ its variable set. 
That is, $B = C_1 \land C_2 \land \dots \land C_m$ is a conjunction of clauses $C_i$, where a clause is a disjunction of literals and a literal is defined as either $x$ or $\overline{x}$ for a variable $x \in X$.\todo{brauchen wir die Variablennamen?}
In \textsc{Planar Monotone 3-Sat}, all clauses consist of at most three literals and are either \emph{positive} (they only contain positive literals) or \emph{negative} (they only contain negative literals), and the variable-clause incidence graph can be drawn such that
\begin{enumerate*}[label=(\roman*)]
    \item it is crossing-free;
    \item all variable vertices lie on the $x$-axis;
    \item all positive clause vertices lie above the $x$-axis; and
    \item all negative clause vertices lie below the $x$-axis.
\end{enumerate*}

We construct a vertex-weighted layered graph $G$ whose contact representation closely resembles the rectilinear representation of~$B$.
	To this end, we use gadgets to represent variables and clauses, as well as an additional gadget to split/duplicate variable values. Just as in the rectilinear representation, vertices representing variable gadgets are aligned horizontally, and positive clauses are drawn above, while negative clauses are drawn below the variable gadgets. The goal is for $G$ to have a valid contact representation if and only if $B$ is satisfiable. We choose $k$ as the maximum number of possible contacts in our construction.

\paragraph{Variable gadget.}
A variable gadget consists of five vertices $v_{l_1}, v_{l_2}, v_{m}, v_{r_2}, v_{r_1}$
that each have a rectangle width of~$1$ on layer $1$, as well as three vertices 
$u_{l}, u_{m}, u_{r}$ on layer $0$. The rectangles $R(u_l)$ and $R(u_r)$ both have 
width~$2$, $R(u_m)$ has width~$1$; see \cref{img:variable_gadget}.
As edges between the layers we add $u_{l}v_{l_1}, u_{l}v_{l_2}, u_{l}v_{m}, u_{l}v_{r_2}$,
$u_{m}v_{r_2}$, $u_{r}v_{r_1}$, and $u_{r}v_{r_2}$.
Note that there is no edge between $v_{m}$ and $u_{m}$, and the corresponding rectangles are therefore not allowed to touch. We want to use this to create a gap in each layer, which will allow us to assign opposite variable values above and below the gadget, thus realizing the notion of positive clauses above and negative clauses below the variable gadgets.

For the gadget to work as intended we need additional walls on either side. Walls are constructed from three rectangles of width $1$ per layer. Edges are added in such a way that moving any wall rectangle to either side reduces realized contacts by at least one and/or introduces false adjacencies; see \cref{img:walls}. 

To determine variable values, we add vertices $v_{x}$ and $u_{x}$ of rectangle width $3$ to layers $2$
and $-1$, respectively, with edges to all vertices of the variable gadget and the innermost wall vertices on the adjacent layers. 
Since $u_{m}$ and $v_{m}$ are not allowed to touch, they split 
the rectangles on layers $0$ and $1$ into two blocks of rectangles of width~$3$ 
and~$2$, respectively. To maximize contacts, both $v_{x}$ and $u_x$ 
have to realize vertical contacts to the larger block of width 3 and
a horizontal contact to a wall vertex. Since the blocks of width $3$ on 
layers $0$ and~$1$ are in contact with opposite walls, so are $v_{x}$ and 
$u_x$. We interpret a variable assignment as follows: if $v_{x}$
 realizes contacts to $v_{l_1}$ and $v_{l_2}$, and $u_x$ realizes a 
contact to $u_{r}$, the assigned value of the variable is \true, otherwise \false; 
see \cref{img:variable_gadget-true,img:variable_gadget-true_graph,img:variable_gadget-false_graph}.

Note that $u_{l}$ could also realize contacts to $v_{l_2}, v_{m}, u_{m}$ 
instead of $v_{l_1}, v_{l_2}$ and a wall vertex; see \cref{img:variable_gadget-alt_1}. However, this does not change the position of $u_x$ and can therefore be 
disregarded. The same holds for vertices $v_{r_2}, v_{r_1}$, which could be 
moved to the left by one without changing the number of realized contacts; see
\cref{img:variable_gadget-alt_2}. Every other valid placement of 
vertices results in the variable gadget to be wider and thus realize less 
contacts; see for example \cref{img:variable_gadget-fewer}.

\begin{figure}[!t]
	\centering
	\subcaptionbox{\label{img:variable_gadget-true}}{\includegraphics[page=2]{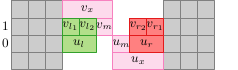}}\hfill
	\subcaptionbox{\label{img:variable_gadget-true_graph}}{\includegraphics[page=3]{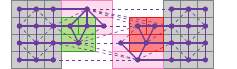}}\hfill	
	\subcaptionbox{\label{img:variable_gadget-false_graph}}{\includegraphics[page=4]{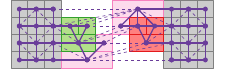}}
	
	\subcaptionbox{\label{img:variable_gadget-alt_1}}{\includegraphics[page=5]{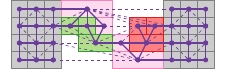}}
	\hfill
	\subcaptionbox{\label{img:variable_gadget-alt_2}}{\includegraphics[page=7]{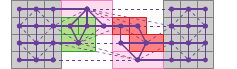}}\hfill	
	\subcaptionbox{\label{img:variable_gadget-fewer}}{\includegraphics[page=8]{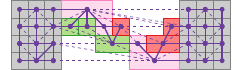}}	
	
	\caption{(a) Contact representation and (b) underlying graph for a variable gadget with variable assignment \true; 
	(c) variable gadget with variable assignment \false;
	(d,e) alternative representations with the same number of realized contacts;
	(f) valid representation realizing fewer contacts.}
	\label{img:variable_gadget}
\end{figure}

In order to use variable values within multiple clauses, we will have to 
propagate them; see \cref{img:propagate}. We do so by adding alternating rows of 
five vertices with rectangle width $1$ and rows of three vertices of width $2$, 
$1$, and $2$, essentially repeating the pattern we used for the variable 
gadget. The difference is that this time the middle rectangles have edges to 
their counterparts in adjacent rows and are therefore allowed to touch. Thus, 
the gap stays as assigned by the variable gadget. We can proceed to add 
vertices $v_{x}$ and $u_x$ as before. 

\begin{figure}[!t]
	\centering
	\subcaptionbox{\label{img:walls}}{\includegraphics{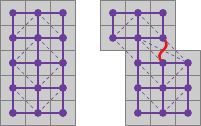}}
	\hfill
	\subcaptionbox{\label{img:propagate}}{\includegraphics{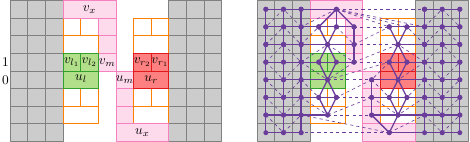}}
	\caption{(a) Moving wall vertices leads to false adjacencies (red curve) and (b) propagating variable values.}
\end{figure}

\paragraph{Clause gadget.}
Let $C$ be a clause that contains variables $x_a,x_b,x_c$ in $B$.
Recall that all clauses above the variable layer are positive while all clauses below the 
variable layer are negative, and the variable gadgets propagate the positive
variable assignment to the top and the negated variable assignment to the bottom.

Assume that $x_a, x_b, x_c$ occur in this order. To determine whether a clause is satisfied, we use a \emph{slider} vertex $v_s$.
The slider shall realize 4 contacts if the variable assignments 
satisfy $C$, and 3 contacts otherwise. 
The slider has rectangle width $2$ and can therefore only %
be in contact 
with one variable gadget at a time.

We describe the clause gadget for the case that $C$ is a negative clause; see \cref{img:clause_gadget}. 
The other case is symmetric. Suppose that the propagation of the variable assignments for 
$x_a, x_b, x_c$ ends with vertices $u_{a}, u_{b}, u_{c}$ on layer $i$.
On layer $i-1$, we place $v_s$ and continue the outermost walls with two vertices 
of rectangle width 3. 
On layer $i-2$, we add vertices 
$v_{l_1}, v_{l_2}, v_{t_1}, v_{b_1}, v_{t_2}, v_{b_2}, v_{t_3}, v_{b_3}, v_{r_1}, v_{r_2}$
in this order to close the bottom of the gadget. 
The rectangles $R_{l_1}, R_{l_2}, R_{r_1}, R_{r_2}$ have width~$1$; 
$R_{t_1}, R_{t_2}, R_{t_3}$ have width $2$. The width of 
$R_{b_1}, R_{b_2}, R_{b_3}$ is set such that the remaining space is filled and 
$v_{t_1}, v_{t_2}, v_{t_3}$ are each placed on the leftmost position 
underneath a variable gadget, i.e., on the side of the positive-valued 
variable propagation. Edges exist from $v_s$ to most vertices of the gadget 
on adjacent layers %
such that the triangulation is preserved and 
$v_s$ can be placed freely along the whole width of the gadget. For the exact edges, refer to \cref{img:clause_gadget}.

The only ways for $v_s$ to realize four contacts are the following.
\begin{enumerate*}[label=(\roman*)]
	\item it touches $v_{t_1}$ and $v_{b_1}$ at the bottom, the wall at the left, and $u_{a}$ at the top, if $x_a$ has a negative variable assignment;
	\item it touches $v_{b_1}$ and $v_{t_2}$ at the bottom, $u_{b}$ and the wall left of $u_{b}$ at the top, if $x_b$ has a negative variable assignment; or
	\item it touches $v_{b_2}$ and $v_{t_3}$ at the bottom, $u_{c}$ and the wall left of $u_{c}$ at the top, if $x_c$ has a negative variable assignment.
\end{enumerate*}
Thus, $v_s$ only realizes four contacts if the variable assignment satisfies $C$.

\begin{figure}[!t]
	\centering
	\includegraphics{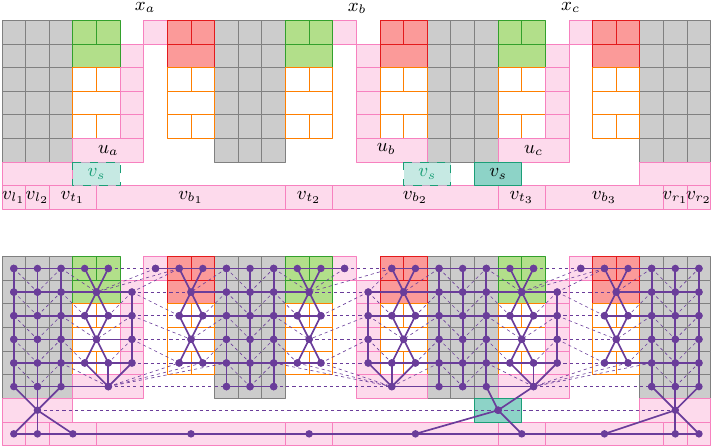}
	\caption{Contact representation (top) and underlying graph (bottom) for a clause gadget, including multiple examples of placements for $v_s$.
	Unrealized edges between $v_s$ and vertices of adjacent layers are omitted for readability.}
	\label{img:clause_gadget}
\end{figure}

\paragraph{Split gadget.}
To duplicate variable values that occur in multiple clauses, we use a split gadget;
see \cref{img:split_gadget}.
Let $x_a$ be a variable such that its variable
assignment ends at a vertex $v_{a}$. 
Recall that all clauses above the variable layer are positive while all clauses below
the variable layer are negative. Assume that $v_{a}$ lies above the variable layer, so the 
variable assignment has to be propagated to a positive clause; the other case is symmetric. 
In the split gadget, we want to
split the variable assignment of $x_a$ such that there are now two vertices $v'_{a}$ and 
$v''_{a}$ that realize the variable assignment of $x_a$. To this end, we create a second tunnel 
to the right of the tunnel that $v_{a}$ lies in and use a horizontal bar $v_m$ that makes sure 
that $v''_{a}$ must have variable assignment \false if $v_{a}$ has variable assignment 
\false; see \cref{img:split_gadget}. Note that the construction also allows $v''_{a}$ to have 
variable assignment \false if $v_{a}$ has variable assignment \true. 
However, this is not a 
problem since it will propagate the variable assignment to a positive clause; hence, this cannot
satisfy a clause that should be unsatisfied due to the variable assignment.
 More details are given in 
\cref{app:split_gadget}. 

\begin{figure}[!t]
	\centering
		\subcaptionbox{\label{img:split_gadget-contact}}{\includegraphics[page=2]{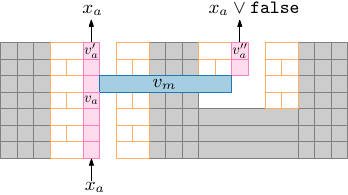}\vspace{-1.5em}}
    \hfill
    \subcaptionbox{\label{img:split_gadget-graph}}{\includegraphics[page=9]{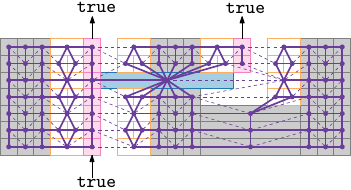}\vspace{-1.5em}}
	\caption{(a) Contact representation and (b) underlying graph for a split gadget.}
	\label{img:split_gadget}
\end{figure}

\paragraph{Combining the gadgets.}
Combining these gadgets, we construct a vertex-weighted layered graph $G$ for the planar monotone boolean formula $B$.
Let $h$ be the number of layers of $G$, and let $w$ be the minimum width of $G$ (i.e., the sum of rectangle widths 
among all layers). Obviously, any layered contact representation of $G$ has at most $w\cdot h$ contacts. To make
sure that the representation of $G$ has to be drawn inside a designated bounding box of width $w$ and height $h$, we add a \emph{frame} 
around $G$ consisting of walls of width $w\cdot h$ on the left and right and 
$w\cdot h$ stacked rectangles that span the whole width of $G$ at the top and bottom, creating a graph
$G^+$; see \cref{img:frame} in \cref{app:figures}.
Moving any rectangle of $G$ outside of the designated bounding box also moves parts of the frame and thus removes at 
least $w\cdot h$ contacts. We choose the number of desired contacts $k$ as the number of contacts that would be 
realized if every single gadget maximizes its number of contacts.
A full example can be seen in \cref{img:example_formula}.

Assume that we have a solution for $B$. For each variable, we draw the corresponding variable gadget of $G^+$ such 
that it represents the variable assignment of the solution, and we propagate the variable assignments along the
tunnels and split gadgets. Since the variable assignment satisfies all clauses, we can place $v_s$ at each clause
such that it has 4 contacts, thus maximizing the number of contacts at every gadget and obtaining $k$ contacts in total.

For the other direction, assume that we have a drawing of $G^+$ that realizes $k$ contacts. From each variable gadget, 
we can read the corresponding variable assignment. Since each clause gadget must have $v_s$ in a position
such that it has four contacts (otherwise, there cannot be $k$ contacts in total), every clause has
a satisfied literal. Together with \cref{lem:lp-crown-np}, this proves the following theorem.

\begin{theorem}
	$k$-\ilc is NP-complete for internally triangulated graphs.
\end{theorem}

\begin{figure}[t]
	\centering
	\includegraphics[width=\textwidth]{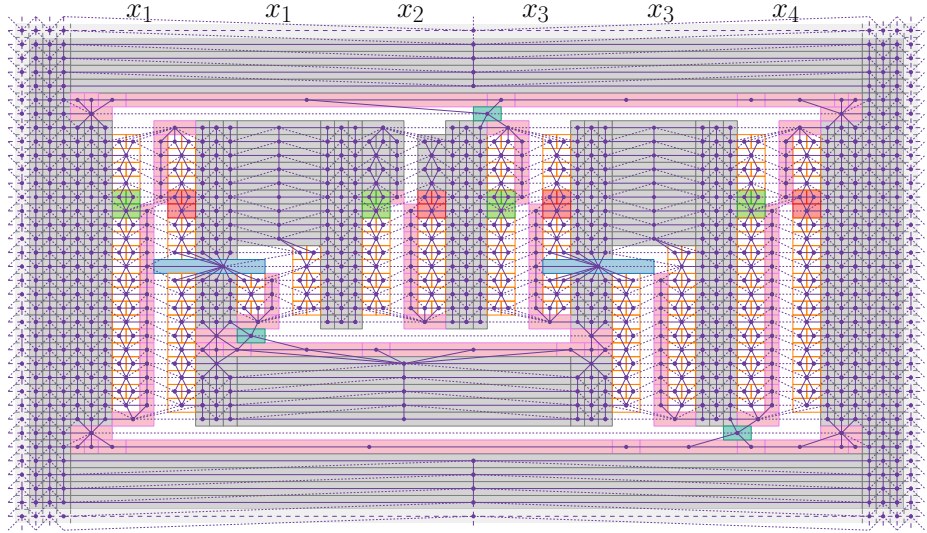}
	\caption{Contact representation for the boolean formula $B$ with variable set $\{x_1, x_2, x_3, x_4\}$, clauses $\{\overline{x_1},\overline{x_2}, \overline{x_3}\}, \{\overline{x_1}, \overline{x_3}, \overline{x_4}\}$ and $\{x_1, x_3, x_4\}$ and variable assignment $x_1 = \false, x_2 = \true, x_3 = \true, x_4 = \false$ (Edges between $v_s$ and above/below layer left out for readability purposes). A larger version can be found in \cref{app:figures}.}
	\label{img:example_formula}
\end{figure}

Note that the proof cannot be immediately extended to $k$-\lc, as placing rectangles on non-integer positions might lead to situations where a variable assignment flips; see \cref{img:nonint-counterex}. However, if we drop the requirement that the graph is triangulated, then we can adjust the construction by removing unwanted contacts from the graph, which leads to the following theorem. The details are given in \cref{app:klc-np-complete}.

\begin{figure}[t]
    \centering
    \includegraphics{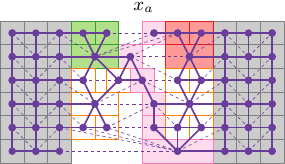}
    \caption{If non-integral positions are allowed, then variable assignments may flip.}
    \label{img:nonint-counterex}
\end{figure}

\begin{restatable}[\restateref{thm:lp-crown-np}]{theorem}{LPCrownNP}
    \label{thm:lp-crown-np}
	$k$-\lc is NP-complete.
\end{restatable}

\section{Parameterized and Approximation Algorithms}\label{ch:algorithms}
In this section, we provide parameterized and approximation algorithms. As a warmup (\cref{sec:2-approx}), we first describe a 1/2-approximation for \mlc on triangulated graphs. 
We then focus on \milc with the additional constraint that the maximum rectangle width is at most polynomial in $n$. Note that practical instances of \milc will always have bounded maximum rectangle width, as each rectangle corresponds to a word, and words have an upper limit of letters in most languages (in fact, the longest word in an English dictionary, has 45 letters: \emph{pneumonoultramicroscopicsilicovolcanoconiosis}).
We first describe an XP-algorithm based on a dynamic program (\cref{sec:dp-alg}), which we then use to obtain a PTAS (\cref{sec:ptas}).

\subsection{1/2-Approximation Algorithm for \mlc}\label{sec:2-approx}

We show that a 1/2-approximation exists by describing an algorithm that uses the following Lemma, proposed by Nöllenburg et al.~\cite{Noellenburg:21}.
\begin{lemma}[\cite{Noellenburg:21}, Theorem 2\label{lem:2-layer-alg}]
	A contact-maximal valid representation 
    for a given triangulated 2-layer graph can be computed in linear time.
\end{lemma}

In the following theorem, we split a $k$-layer graph into many 2-layer graphs and solve these optimally
with \cref{lem:2-layer-alg}. Half of these 2-layer graphs are vertex-disjoint, so their optimal solutions can be 
combined to a valid solution of the input graph.

\begin{theorem}\label{thm:1/2-approx}
	\mlc on triangulated graphs admits a 1/2-appro\-xi\-mation in linear time.
\end{theorem}
\begin{proof}
    Let $G$ be an $L$-layered graph. 
    For $i=1,\ldots, L-1$, let $A_i$ be the subgraph of $G$ induced by the vertices on layers $i$ and $i+1$.
    We construct two groups of subgraphs $G_\mathrm{even}=\bigcup_{i \text{ even}} A_i$ and 
    $G_\mathrm{odd}=\bigcup_{i \text{ odd}} A_i$; see \cref{img:2-approx}.
    
    We solve every subgraph $A_i$, $i=1,\ldots, L-1$ optimally using \cref{lem:2-layer-alg}.
    Let $\ALG_i$ be the number of contacts realized for $A_i$. Let $\Gamma^*$ be an optimal drawing of $G$ that
    realizes $\OPT$ contacts, and let $\OPT_i$ be the number of contacts realized for $A_i$ in $\Gamma^*$.
	Since the 2-layer algorithm yields an optimal solution, it holds that $\ALG_i \ge \OPT_i$ for $i=1,\ldots,L-1$, so
	$\sum_{i=1}^{L-1}\ALG_i \ge \sum_{i=1}^{L-1}\OPT_i \geq \OPT$.
    Note that any two subgraphs $A_i,A_j\in G_\mathrm{even}$ are vertex-disjoint. Hence, we can obtain a valid solution for
    $G_\mathrm{even}$ with $\ALG_{\mathrm{even}} = \sum_{i : A_i \in G_\mathrm{even}}\ALG_i$ contacts by combining the computed solutions for the
    corresponding subgraphs. Analogously, we can obtain a valid solution for $G_\mathrm{odd}$ with 
     $\ALG_{\mathrm{odd}} = \sum_{i : A_i \in G_\mathrm{odd}}\ALG_i$ contacts.
	We get a 1/2-approximation by choosing the contacts realized by the instances corresponding to the larger of both sums: $\max\{\ALG^{G_\mathrm{even}}, \ALG^{G_\mathrm{odd}}\}\geq \OPT/2$. 

    For the running time, note that every vertex lies in at most two subgraphs, and \cref{lem:2-layer-alg} solves each subgraph optimally in time linear in its size.
\end{proof} 

\begin{figure}[!t]
	\begin{minipage}[b]{.47\textwidth}
		\centering
		\includegraphics{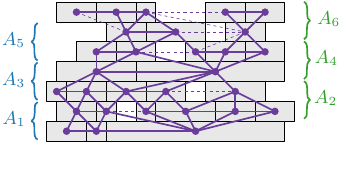}
		\caption{2-layered sub-graphs split into two groups $G_\mathrm{odd}$ (blue) and $G_\mathrm{even}$ (green).}
		\label{img:2-approx}
	\end{minipage}
	\hfill 
	\begin{minipage}[b]{.5\textwidth}
		\centering
		\includegraphics{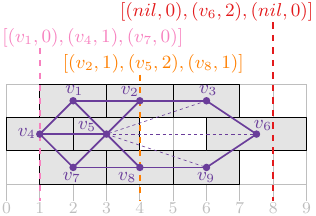}
		{\captionsetup{format=plain}
		\caption{Three cuts and respective $k$-tuples of an optimal solution of a $3$-layered graph.}}
		\label{img:dp-1}
	\end{minipage}
\end{figure}

\subsection{XP-Algorithm for \milc}\label{sec:dp-alg}
We now use a dynamic programming approach to solve \milc with bounded maximum rectangle width optimally.

\begin{theorem}\label{theorem:dp}
	\milc is solvable in time $\mathcal{O}(nW)^L$, where $W$ is the maximum rectangle width. If $W\in\poly(n)$, \milc lies in XP when parameterized by the number of layers $L$ of the input graph.
\end{theorem}
\begin{proof}
Let $G$ be an $L$-layered vertex-weighted graph with maximum weight~$W$. 
We want to define subproblems based on vertical cuts through integer $x$-coordinates; see \cref{img:dp-1}.
At each such cut through any valid representation, we can obtain the following information:
\begin{enumerate*}[label=(\roman*)]
	\item Which vertex has been cut at each layer (if any)?
	\item At what length was the corresponding rectangle cut (i.e., how much of the rectangle has already been drawn on the left of the cut)?
	\item If no vertex has been cut, which vertex will be drawn next on the specific layer?
\end{enumerate*}	

To formalize this, we use a tuple $(v,l)$ for each layer, where $v$ denotes the vertex that is being cut and $l$ denotes the length of the rectangle to the left of the cut. 
The tuple $(v,0)$ indicates that $v$ is next in line but has not yet been placed, while $(nil,0)$ means that there is no more vertex to be drawn on the corresponding layer. 
For every possible cut, we therefore obtain an $L$-tuple $[(v_1, l_1), (v_2, l_2), \ldots, (v_L, l_L)]$.
As each rectangle has at most width $W$, there are no more than $((n+1) \cdot (W+1))^L$ such $L$-tuples.
We store in an $L$-dimensional table $D$ for each $L$-tuple the maximum number of contacts that can be 
achieved to the {\em right} of the corresponding cut.

We set $D[(nil,0),\dots,(nil,0)]=0$, which corresponds to the right boundary of the drawing.
Consider any $L$-tuple $T$ and its corresponding cut. To calculate $D[T]$, we have to look at each cut $T'$ through
a solution one coordinate to the right. Consider any layer $i$ and the corresponding tuple $T_i=(v_i,l_i)$; see \cref{img:dp-2}.
If $T$ cuts through the middle of $v_i$, i.e., $0<l_i<w(v_i)$, then this rectangle has to continue, i.e., $T'_i=(v_i,l_i+1)$. If $T$ cuts through no vertex, i.e., $l_i=0$, then we can either place $v_i$, i.e., $T'_i=(v_i,1)$, or not place it yet, i.e., $T'_i=(v_i,0)$. Finally, if $T$ touches the right side of $v_i$, i.e., $l_i=w(v_i)$, then we can either immediately place the next vertex $v_i'$ (if it exists), i.e., $T'_i=(v_i',1)$, or not place it yet, i.e., $T'_i=(v_i',0)$. Doing this for every layer, we can find each possible next cut. For each such cut, we calculate whether it is feasible, i.e., whether the newly placed vertices have any false adjacencies. If it is not feasible, then we discard it; otherwise, we count how many edges are realized by the newly placed vertices, and thus calculate $D[T]$ from $D[T']$.
We can obtain the optimum solution for $G$ from $D[(v_1,0),\ldots,(v_L,0)]$, where $v_i$ is the leftmost vertex of layer $i=1,\ldots,L$.

\begin{figure}[!t]
	\centering
	\includegraphics[width=\textwidth]{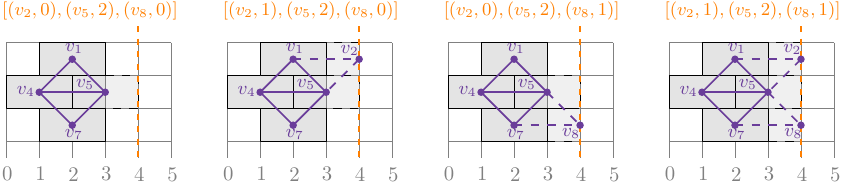}
	\caption{Possible assignments for an $L$-tuple at $x$-coordinate $4$ following an assignment of $[(v_1,2),(v_5,1),(v_7,2)]$ at $x$-coordinate $3$.}
	\label{img:dp-2}
\end{figure}

All in all, this leaves us with at most $(n \cdot (W+1)+1)^L$ different table entries that each take $\Oh(2^L)$ time to be calculated. The algorithm thus runs in $\mathcal{O}(nW)^L$ time.
To obtain the solution instead of the number of contacts, we can use an additional lookup table in the same time.
\end{proof}

\subsection{PTAS for \milc}\label{sec:ptas}

In the following, we use Baker's technique~\cite{Baker:94} to combine the ideas of the previously described 1/2-approximation (\cref{sec:2-approx}) and dynamic program (\cref{sec:dp-alg}).

\begin{lemma}\label{lem:approx}
	For every integer $\ell > 0$, \milc admits a $(1-\frac{1}{\ell})$-ap\-pro\-xi\-ma\-tion in $\Oh(nW)^{\ell+1}$ time, where $W$ is the maximum rectangle width.
\end{lemma}

\begin{proofWithMath}
    If $\ell\ge L$, then we can solve the problem optimally in $\Oh(nW)^L$ time using \cref{theorem:dp}.
	Otherwise, similar to \cref{thm:1/2-approx}, we split the graph into multiple subgraphs of $\ell$ layers each, which we will then solve using the dynamic program described in \cref{theorem:dp}.     We assume that $L$ is evenly divisible by $\ell$; otherwise, we add empty dummy layers to the top, increasing $L$ by at most factor 2. For technical reasons, we treat layer 0 to be the same as layer~$L$.
 
    For $i=1,\ldots, L$, let $A_i$ be the subgraph of $G$ induced by the vertices on the $\ell$ layers $i, \ldots, i+\ell-1 \mod L$. We can solve each of these subgraphs optimally using \cref{theorem:dp} in $\Oh(nW)^\ell$ time. Since $L\in \Oh(n)$, this takes $\Oh(nW)^{\ell+1}$ time in total. Let $\ALG_i$ be the number of contacts for $A_i$ obtained this way.
	
    Let $\Gamma^*$ be an optimal representation of $G$ that realizes $\OPT$ contacts, 
    let $\OPT_i$ be the number of horizontal contacts realized for each layer $i$ in $\Gamma^*$, 
    and let $\OPT_{i,i+1}$ denote the number of vertical contacts between layers $i$ and $i+1$ in~$\Gamma^*$. 
    Since we solved $A_i$ optimally, we have 
    
    \[\ALG_i \geq \sum_{j=i}^{(i+\ell-2)\!\!\!\!\mod L}(\OPT_j + \OPT_{j,j+1}) + \OPT_{i+\ell-1}.\] 

	 Horizontal contacts of each layer are covered by $\ell$ subgraphs and vertical contacts between pairs $\{i,i+1\}$ of layers are covered by $\ell-1$ subgraphs. Therefore,
	\[
	\sum_{i=1}^L \ALG_i \geq \ell \sum_{j=1}^{L}\OPT_j + (\ell-1) \sum_{j=1}^{L-1} \OPT_{j,j+1} \geq (\ell-1)\OPT
	\]

     We then partition these subgraphs into $\ell$ groups $G_1\ldots,G_{\ell}$ such that $G_i = A_i\cup A_{i+\ell}\cup A_{i+2\ell}\cup \ldots \cup A_{i+L-\ell}$; see \cref{img:approx}. Note that the subgraphs in a group are vertex-disjoint,  so combining the optimum solutions for $A_i,A_{i+\ell},\ldots, A_{i+L-\ell}$ gives an optimum solution for $G_i$ with $\ALG_{G_i} = \sum_{j=0}^{ L/\ell-1}\ALG_{i+j\ell}$ contacts. Further, every subgraph lies in exactly one group, so $\sum_{i=1}^{\ell} \ALG_{G_i} = \sum_{i=1}^L \ALG_i$.

     We now choose $1\le j\le \ell$ such that $\ALG_{G_j}= \max_{i=1}^{\ell}\ALG_{G_i}$.
     Then, \\[-1ex]
		\[\ALG_{G_j}= \max_{i=1}^{\ell}\ALG_{G_i} \geq \frac{1}{\ell}\sum^{L}_{i=1}\ALG_{G_i} \geq (1-\frac{1}{\ell})\OPT.\tag*{\qed}\]
\end{proofWithMath}

For any $\varepsilon{>}0$, by choosing $\ell{=}\lceil 1/\varepsilon\rceil$, \cref{lem:approx} provides a PTAS if $W{\in}\poly(n)$.

\begin{theorem}
	For every $\varepsilon >0$, \milc admits a $(1{-}\varepsilon)$-ap\-pro\-xi\-ma\-tion in $\Oh(nW)^{1+\lceil\frac{1}{\varepsilon}\rceil}$ time, where $W$ is the maximum rectangle width.
\end{theorem}

\begin{figure}[!t]
	\centering
	\includegraphics{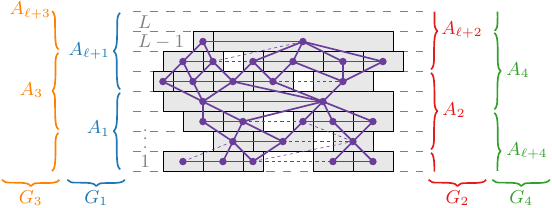}
	\caption{An $L$-layered graph with one dummy layer split into $L$ subgraphs of $\ell$ layers, partitioned into $\ell$ groups, for $L=8$ and $\ell=4$.}
	\label{img:approx}
\end{figure}

\section{Conclusion}\label{ch:discussion}

We have proved that $k$-\ilc and $k$-\lc are NP-complete, and provided an XP-algorithm parameterized by the number of layers and a PTAS for \milc when rectangle widths are polynomial in $n$. 
Several interesting problems remain open, for example:
\begin{enumerate*}[label=(\arabic*)]
    \item Is there an FPT-algorithm parameterized by the number of layers for \milc?
    \item Is there a PTAS for \milc for which the running time does not depend on the maximum rectangle width?
    \item What can we do if rectangles can have different (integer) heights, thus spanning more than one layer? 
\end{enumerate*}

\bibliographystyle{splncs04}
\renewcommand{\doi}[1]{\href{https://dx.doi.org/#1}{\texttt{doi:#1}}}
\bibliography{abbrv,literatur}

\begin{thebibliography}{10}
\providecommand{\url}[1]{\texttt{#1}}
\providecommand{\urlprefix}{URL }
\providecommand{\doi}[1]{https://doi.org/#1}

\bibitem{Avidan2007}
Avidan, S., Shamir, A.: Seam carving for content-aware image resizing. {ACM}
  Trans. Graphics  \textbf{26}(3),  10--1--10--9 (2007).
  \doi{10.1145/1276377.1276390}

\bibitem{Berg1997}
de~Berg, M., van Kreveld, M., Overmars, M., Schwarzkopf, O.: Computational
  Geometry. Springer (1997)

\bibitem{Buchin2016}
Buchin, K., Creemers, D., Lazzarotto, A., Speckmann, B., Wulms, J.: Geo word
  clouds. In: Proc. 9th {IEEE} Pacific Vis. Symp. ({PacificVis}'16). {IEEE}
  (2016). \doi{10.1109/pacificvis.2016.7465262}

\bibitem{Cui2010}
Cui, W., Wu, Y., Liu, S., Wei, F., Zhou, M.X., Qu, H.: {C}ontext-{P}reserving,
  {D}ynamic {W}ord {C}loud {V}isualization. Comput. Graph. Appl.
  \textbf{30}(6),  42--53 (2010). \doi{10.1109/mcg.2010.102}

\bibitem{Cygan2015}
Cygan, M., Fomin, F.V., Kowalik, {\L}., Lokshtanov, D., Marx, D., Pilipczuk,
  M., Pilipczuk, M., Saurabh, S.: Parameterized Algorithms. Springer (2015).
  \doi{10.1007/978-3-319-21275-3}

\bibitem{Downey:13}
Downey, R.G., Fellows, M.R.: Fundamentals of Parameterized Complexity, TCS,
  vol.~4. Springer (2013). \doi{10.1007/978-1-4471-5559-1}

\bibitem{Fomin:19}
Fomin, F.V., Lokshtanov, D., Saurabh, S., Zehavi, M.: Kernelization: Theory of
  Parameterized Preprocessing. Cambridge University Press (2019)

\bibitem{Lawler1979}
Lawler, E.L.: {F}ast {A}pproximation {A}lgorithms for {K}napsack {P}roblems.
  Mathematics of Operations Research  \textbf{4}(4),  339--356 (1979).
  \doi{10.1287/moor.4.4.339}

\bibitem{Li2018}
Li, C., Dong, X., Yuan, X.: Metro-wordle: An interactive visualization for
  urban text distributions based on wordle. Visual Informatics  \textbf{2}(1),
  50--59 (2018). \doi{10.1016/j.visinf.2018.04.006}

\bibitem{Wang2020}
Wang, Y., Lee, B., Chu, X., Zhang, K., Bao, C., Li, X., Zhang, J., Fu, C.W.,
  Hurter, C., Deussen, O.: {ShapeWordle}: Tailoring wordles using shape-aware
  archimedean spirals. {IEEE} Trans. Visual. Comput. Graphics  \textbf{26}(1),
  991--1000 (2020). \doi{10.1109/tvcg.2019.2934783}

\bibitem{Wu2011}
Wu, Y., Provan, T., Wei, F., Liu, S., Ma, K.L.: Semantic-preserving word clouds
  by seam carving. Comput. Graph. Forum  \textbf{30}(3),  741--750 (2011).
  \doi{10.1111/j.1467-8659.2011.01923.x}

\end{thebibliography}


\begin{thebibliography}{10}
\providecommand{\url}[1]{\texttt{#1}}
\providecommand{\urlprefix}{URL }
\providecommand{\doi}[1]{https://doi.org/#1}

\bibitem{Baker:94}
Baker, B.S.: Approximation algorithms for {NP}-complete problems on planar
  graphs. J. ACM  \textbf{41}(1),  153--180 (1994). \doi{10.1145/174644.174650}

\bibitem{Barth:14}
Barth, L., Fabrikant, S.I., Kobourov, S.G., Lubiw, A., N{\"{o}}llenburg, M.,
  Okamoto, Y., Pupyrev, S., Squarcella, C., Ueckerdt, T., Wolff, A.: Semantic
  word cloud representations: Hardness and approximation algorithms. In: Pardo,
  A., Viola, A. (eds.) Proc. 11th Latin Am. Symp. Theoretical Informatics
  ({LATIN}'14). Lecture Notes Comput. Sci., vol.~8392, pp. 514--525. Springer
  (2014). \doi{10.1007/978-3-642-54423-1\_45}

\bibitem{Barth2014}
Barth, L., Kobourov, S.G., Pupyrev, S.: Experimental comparison of semantic
  word clouds. In: Gudmundsson, J., Katajainen, J. (eds.) Proc. 13th Int. Symp.
  Experimental Algorithms ({SEA}'14). Lecture Notes Comput. Sci., vol.~8504,
  pp. 247--258. Springer (2014). \doi{10.1007/978-3-319-07959-2\_21}

\bibitem{DiBattista:99}
Battista, G.D., Eades, P., Tamassia, R., Tollis, I.G.: Graph Drawing:
  Algorithms for the Visualization of Graphs. Prentice-Hall (1999)

\bibitem{Bekos2014}
Bekos, M.A., van Dijk, T.C., Fink, M., Kindermann, P., Kobourov, S.G., Pupyrev,
  S., Spoerhase, J., Wolff, A.: Improved approximation algorithms for box
  contact representations. Algorithmica  \textbf{77}(3),  902--920 (2017).
  \doi{10.1007/s00453-016-0121-3}

\bibitem{DeBerg:12}
de~Berg, M., Khosravi, A.: Optimal binary space partitions for segments in the
  plane. Int. J. Comput. Geom. Appl.  \textbf{22}(3),  187--206 (2012).
  \doi{10.1142/S0218195912500045}

\bibitem{Cygan2015}
Cygan, M., Fomin, F.V., Kowalik, {\L}., Lokshtanov, D., Marx, D., Pilipczuk,
  M., Pilipczuk, M., Saurabh, S.: Parameterized Algorithms. Springer (2015).
  \doi{10.1007/978-3-319-21275-3}

\bibitem{Downey:13}
Downey, R.G., Fellows, M.R.: Fundamentals of Parameterized Complexity, TCS,
  vol.~4. Springer (2013). \doi{10.1007/978-1-4471-5559-1}

\bibitem{Espenant:22}
Espenant, J., Mondal, D.: Streamtable: An area proportional visualization for
  tables with flowing streams. In: Mutzel, P., Rahman, M.S., Slamin (eds.)
  Proc. 16th Int. Workshop Algorithms and Computation ({WALCOM'22}). Lecture
  Notes Comput. Sci., vol. 13174, pp. 97--108. Springer (2022).
  \doi{10.1007/978-3-030-96731-4\_9}

\bibitem{Noellenburg:21}
N{\"{o}}llenburg, M., Villedieu, A., Wulms, J.: Layered area-proportional
  rectangle contact representations. In: Purchase, H.C., Rutter, I. (eds.)
  Proc. 29th Int. Symp. Graph Drawing and Network Visualization ({GD}'21).
  Lecture Notes Comput. Sci., vol. 12868, pp. 318--326. Springer (2021).
  \doi{10.1007/978-3-030-92931-2\_23}

\bibitem{Tamassia:13}
Tamassia, R. (ed.): Handbook on Graph Drawing and Visualization. Chapman and
  Hall/CRC (2013), \url{https://cs.brown.edu/people/rtamassi/gdhandbook}

\bibitem{Viegas2009}
Viegas, F., Wattenberg, M., Feinberg, J.: Participatory visualization with
  wordle. {IEEE} Trans. Visual. Comput. Graphics  \textbf{15}(6),  1137--1144
  (2009). \doi{10.1109/tvcg.2009.171}

\end{thebibliography}

\clearpage
\appendix

\noindent{\large\bf Appendix}
\section{Basic Definitions}\label{ch:prelim}

In this section, we will review some basic concepts of graph drawing (\cref{sec:graphs}) and parameterized complexity (\cref{sec:paracompl}).

\subsection{Graphs and their Contact Representations}\label{sec:graphs}

A drawing $\Gamma$ of a graph $G = (V,E)$ is a mapping that assigns each vertex $v \in V$ a point in $\mathbb{R}^2$, and each edge $\{u,v\} \in E$ a simple open curve $\Gamma(u,v)$ with endpoints $\Gamma(u)$ and $\Gamma(v)$.
A drawing of a graph is called \emph{planar} if there are no intersections of distinct edges. A graph that admits a planar drawing is called \emph{planar}.
For each vertex in a planar drawing, the clockwise order of incident edges is fixed.
Such a clockwise ordering of edges around vertices defines a \emph{planar embedding}. Each planar embedding can admit multiple planar drawings. Two planar drawings with the same embedding are called equivalent. The connected regions that the edges of a planar graph divide the plane into are called \emph{faces}. The outermost, unbounded face is called \emph{outer face}, all other faces are \emph{inner faces}. If every inner face is a triangle, the graph is called \emph{internally triangulated}. See \cref{img:contactrepr-graph} for an example of an internally triangulated graph.
A \emph{layered graph drawing} is a type of graph visualization technique where the vertices of a graph are arranged in 
horizontal layers, and the edges are drawn as segments or curves connecting the vertices~\cite{Tamassia:13}. The goal is 
usually to minimize the number of edge crossings.

In a \emph{contact representation} of a graph, vertices are depicted as geometric objects and two such objects touch, i.e., their boundaries intersect, if and only if there is an edge between their corresponding vertices\cite{Noellenburg:21}. \cref{img:contactrepr} shows a graph with ten vertices and a contact representation where each vertex is drawn as a rectangle. 
\begin{figure}[!t]
	\centering
	\subcaptionbox{\label{img:contactrepr-graph}}{\includegraphics[page=1]{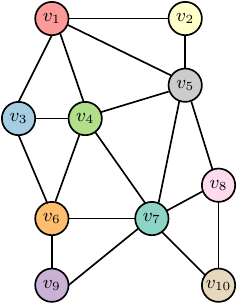}}
    \hfill
    \subcaptionbox{\label{img:contactrepr-rep}}{\includegraphics[page=2]{images/contact.pdf}}
	\caption{(a) An internally triangulated graph and (b) a rectangle contact representation.}
	\label{img:contactrepr}
\end{figure}

\subsection{Parameterized complexity}\label{sec:paracompl}
Many problems that are relevant in practice are NP-complete, and thus, no polynomial-time algorithms are known to solve them. As we show in \cref{ch:complexity}, the problem of drawing layered word clouds is NP-complete, too.
It is only natural to consider different strategies to make those problems more feasible, at least for instances that meet certain criteria. 

Let $\Pi$ be an NP-hard problem. In the framework of parameterized complexity, each instance of $\Pi$ is associated with a {\em parameter} $k$. Here, the goal is to confine the combinatorial explosion in the running time of an algorithm for $\Pi$ to depend only on $k$. Formally, we say that $\Pi$ is {\em fixed-parameter tractable (FPT)} if any instance $(I, k)$ of $\Pi$ is solvable in time $f(k)\cdot |I|^{\Oh(1)}$, where $f$ is an arbitrary computable function of $k$. A weaker request is that for every fixed $k$, the problem $\Pi$ would be solvable in polynomial time. Formally, we say that $\Pi$ is {\em slice-wise polynomial (XP)} if any instance $(I, k)$ of $\Pi$ is solvable in time $f(k)\cdot |I|^{g(k)}$, where $f$ and $g$ are arbitrary computable functions of $k$. 
For more information on parameterized complexity, we refer to books such as \citeapp{Cygan2015,Downey:13,Fomin:19}.

\section{Related Work on Word Clouds}\label{ch:related}

\subsection{Different kinds of word clouds}\label{sec:kinds-of-wordles}
Besides classical and semantic word clouds, other approaches to design word clouds have been proposed. Their aim is to improve the readability, decrease the potential for misinterpretation or add additional information. In the following we discuss three such variants of word clouds: ShapeWordles, which try to fit words into a given shape, geo word clouds, which do the same, but use maps as the shapes to fit, and additionally place words such that they are semantically related to their placement on the map, as well as Metro-Wordles, which arrange word clouds on metro maps, to again convey spatial information.

\subsubsection{Shape Wordle}
Shape Wordles are word clouds in which the words are arranged in such a way that they fit a given shape. See Figure \ref{img:teaparty} (left) for an example of a word cloud generated with \url{ https://wordart.com/}, using again the first chapter of ``Alice's Adventures in Wonderland'' and fitting the words to the shape of a teacup.
The difficulty with those kinds of word clouds lies in appropriately filling a shape, as to make it easily recognizable and visually pleasing while adhering to the desired word sizes given by the frequency of the respective words. 
In the given example those problems were avoided by placing a word multiple times in different sizes. This though allows for easy misinterpretation.

Wang et al. \citeapp{Wang2020} thus introduce a technique to generate ShapeWordles that aim to resolve the aforementioned difficulties in a more semantically consistent manner. In their approach, they use a shape-aware Archimedean spiral to determine word placement. They also consider a multi-centric layout, i.e. shapes that consist of multiple components (Figure \ref{img:teaparty} (right)). Here, too, they use an Archimedean spiral to fill each component. Given non-convex components this might not lead to satisfying results, as large parts of a component might be left blank. They thus consider multiple centers for such shapes, basically segmenting a component into multiple parts. Each part then gets assigned some of the words supposed to fill the given shape greedily and Archimedean spirals are calculated accordingly.

Another area of interest with ShapeWordles is to place words such that their placement aligns with their semantics. In the following, we discuss geo word clouds and Metro-Wordles, which aim to do exactly that.

\begin{figure}[!t]
	\centering
	\includegraphics[width=\linewidth]{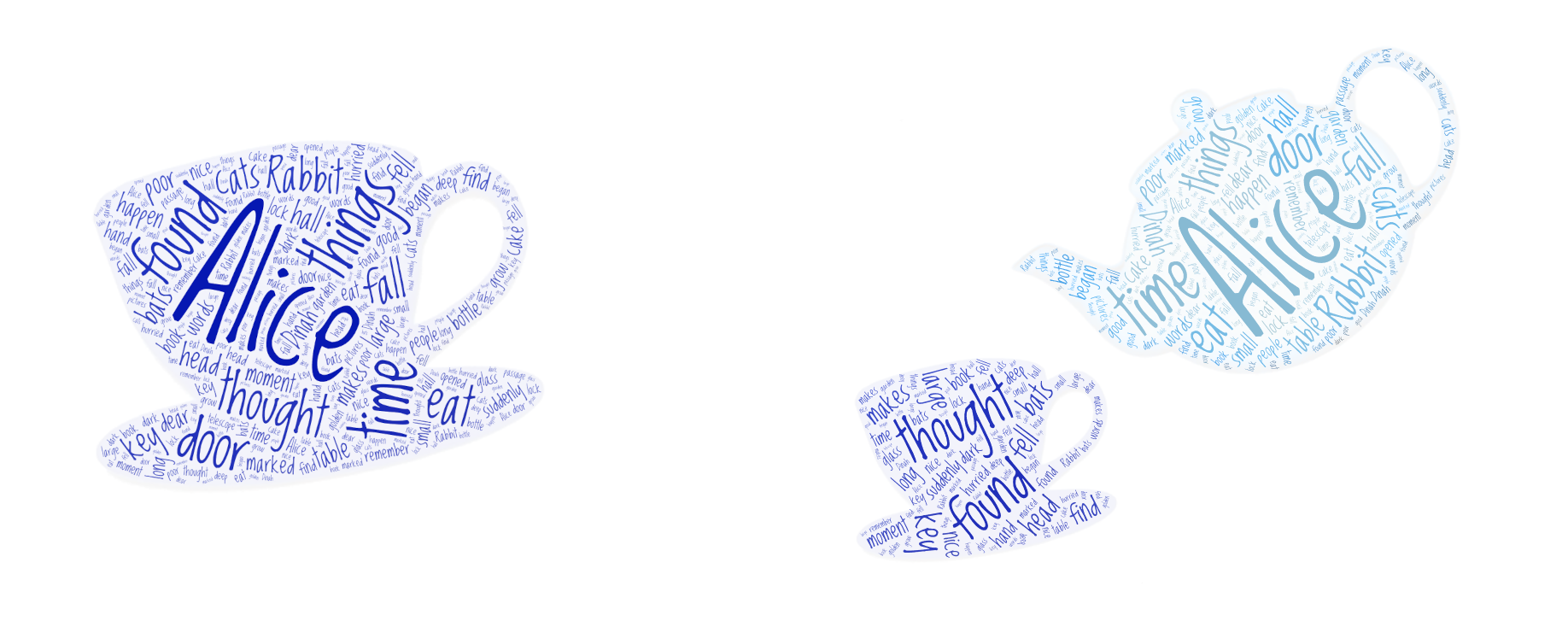}
	\caption{ShapeWordle created using the first chapter of ``Alice's Adventures in Wonderland'' (left), ShapeWordle using the same text, but a shape with multiple components (right).}
	\label{img:teaparty}
\end{figure}

\subsubsection{Geo Word Clouds}
Other than information about word frequency, geo word clouds also depict spatial information. They were introduced in \citeapp{Buchin2016} and, as the name suggests, are aimed at visualizing geographic regions as word clouds (Figure \ref{img:geowordle}). That is, words are linked to points inside a geographic region $M$ and the goal is to draw a word cloud in such a way that words are displayed as closely as possible to their assigned points. Multiple words can be assigned the same point, and words can be assigned multiple points.

At the same time, the words are to be placed in such a way that they ``draw'' the specific region, e.g. if the geographic region is a country, the geo word cloud will resemble the respective shape of said country.
Apart from that, the size of a word still measures its frequency, and words can be rotated to a user-specified degree. 
Buchin et. al \citeapp{Buchin2016} also use color to prohibit misinterpretation of different colored words. If the same word appears multiple times in the word cloud, it will have the same color at every occurrence.
Furthermore, colors are chosen from a Hue-Chroma-Luminance (HCL) color model to ensure that color doesn't give the illusion of importance. Colors are, however, not used to group words semantically or bear any other meaning.

The proposed algorithm works in a greedy fashion. Points that are assigned the same word are clustered, and for each cluster, a single shape of the word is created. The size of the word is chosen in such a way that it covers the cluster well. To place the words, they are then sorted according to their size, and bigger words are placed first.

\begin{figure}[!t]
	\centering
	\includegraphics[scale=1]{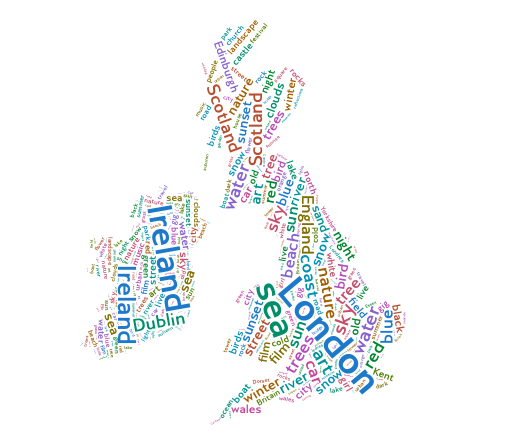}
	\caption{Geo word cloud of Flickr tags in Great Britain and Ireland (\protect\citeapp{Buchin2016}, Figure 8).}
	\label{img:geowordle}
\end{figure}

\subsubsection{Metro-Wordle}
This variant combines traditional word clouds and metro maps \citeapp{Li2018}. Like geo word clouds, Metro-Wordles are therefore spatially informative.

In a Metro-Wordle, a metro map is used as a canvas to depict word clouds containing keywords related to the given city (Figure \ref{img:metrowordle}).
Metro lines serve as a divider, splitting a city into multiple areas, each of which containing word clouds that were generated using keywords from descriptions, reviews, etc. about points of interest (POI) in the given area.

There are four steps in generating a Metro-Wordle. First, as with all word clouds, the displayed (key-)words need to be extracted. Keywords can be not only names of locations, but also sights and things like region-specific food.
In a second step, the metro map has to be drawn such that the topological structure of metro stations is kept and sub-regions (for which word clouds should be computed) have to be defined. The next step is to generate word clouds and place them in the designated areas. As with geo word clouds, words should preferably be placed relative to their actual location on the map. There is, however, some leeway with this, to make sure that word clouds fit into their assigned regions.
Lastly, \citeapp{Li2018} designed an interactive visualization for users to explore a given city via a Metro-Wordle.

\begin{figure}[!t]\label{fig:metrowordle}
	\centering
	\includegraphics[scale=1]{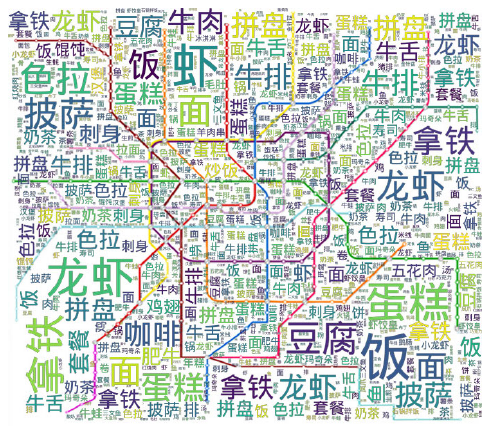}
	\caption{Shanghai Metro-Wordle (\protect\citeapp{Li2018}, Figure 11)}
	\label{img:metrowordle}
\end{figure}

\subsection{Algorithmic Approaches for Computing Semantic Word Clouds}\label{sec:algo-wordles}
This section describes different algorithms to draw semantic word clouds. All of them use some force-directed model at some point, most of them in post-processing.
The last two Sections propose approximation algorithms, the others follow heuristic approaches.

\subsubsection{Context-Preserving (Dynamic) Word Cloud Visualization}\label{sec:cpwcv}
This approach, described in \citeapp{Cui2010}, was introduced to illustrate content evolution in a set of documents over time. It consists of two components, namely a \emph{trend chart viewer} and a \emph{word cloud generator}, of which we will consider only the latter.

Words are initially placed based on a dissimilarity matrix $\Delta$, where each entry $\delta_{i,j}$ describes the similarity between words $i$ and $j$ using $n$-dimensional feature vectors. Depending on the criterion chosen, $n$ denotes either the number of time points of the documents (Importance criterion, Co-Occurrence criterion) or the total number of words (Similarity criterion).
Using multidimensional scaling (MDS) each $n$-dimensional vector is then reduced to a $2$-dimensional point.

After the initial placement, a force-directed model is applied to reduce white space while keeping semantic relations between words. For this, a triangulated mesh is constructed via a \emph{Delaunay triangulation} \citeapp{Berg1997}, i.e. a triangulated graph $G=(V,E)$ is created where words correspond to vertices and for each clique of three vertices $v_1, v_2, v_3 \in V$ it holds that within the cycle through $v_1, v_2, v_3$ there lies no other vertex.

An adapted force-directed algorithm then rearranges vertices using attractive forces between vertices to reduce empty space and repulsive forces to prevent overlapping of words. A third force is used to attempt to keep the mesh planar, as to preserve semantic relationships between words. However, given that planarity might waste space, and conversely, non-planarity does not imply that semantic relationships are lost, keeping the mesh planar is not a strict constraint. Figure \ref{img:cpwcv} visualizes the complete pipeline for creating semantic word clouds in this way.

\begin{figure}[!t]\label{fig:cpwcv}
	\centering
	\includegraphics[width=\linewidth]{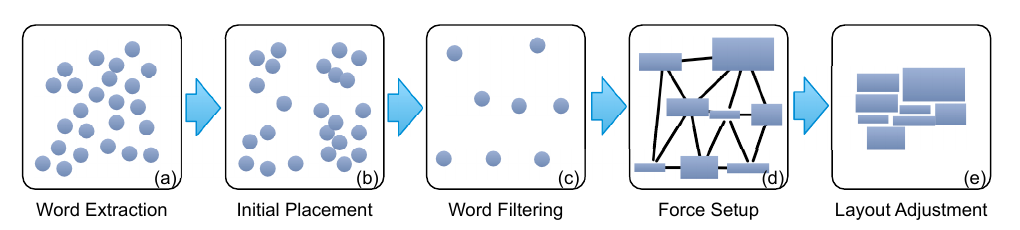}
	\caption{Visualization of the whole pipeline for creating semantic word clouds (\protect\citeapp{Cui2010}, Figure 2); (a) initial set of words extracted from multiple documents with different timestamps; (b) initial placement based on the dissimilarity matrix and using MDS; (c) removing words that are unrelated for a specified timestamp; (d) constructing a triangulated graph; (d) applying the force-directed algorithm.}
	\label{img:cpwcv}
\end{figure}

\subsubsection{Seam carving}\label{sec:seam-carving}
Seam carving describes an algorithm usually used to resize pictures \citeapp{Avidan2007}. ``Seams'', in this case, are horizontal or vertical paths of pixels along the picture that are deemed the least important. When resizing the image to be smaller, the least important seams get removed in each step.

Wu et al. \citeapp{Wu2011} used seam carving to generate semantic word clouds that are more compact and more consistently preserve semantic relations, as opposed to those created with the force-directed approach discussed in Section \ref{sec:cpwcv}.
The algorithm starts by extracting keywords from a collection of texts for which a similarity value is calculated. Based on the calculated values a word cloud is then created. A force-directed approach, focusing only on repulsive forces between words, is used afterward to remove any overlaps of words.
Next, keywords that don't appear in the text, which the word cloud is to be created for, are removed. This likely results in a sparse word cloud and thus seam carving is used to remove unnecessary white space.  The goal is to remove white space in such a way that visualized semantic relations between words remain the same. A so-called \emph{energy function}, based on a Gaussian distribution, is used to determine which semantic relations between words are more important to keep. The energy function is calculated for different regions that the layout is partitioned into using the bounding boxes of the words. A seam then is a horizontal or vertical path of connected regions of low energy. Figure \ref{img:seamcarving} illustrates the removal of a seam. 

The algorithm repeatedly removes seams until there exist no further seams that can be removed.

\begin{figure}[!t]
	\centering
	\includegraphics[width=\linewidth]{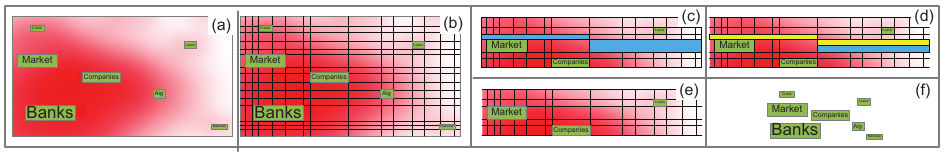}
	\caption{Visualization of the removal of a seam  (\protect\citeapp{Wu2011}, Figure 2); (a) initial layout with Gaussian importance field; (b) Partitioning of the layout using bounding boxes; (c) identifying of a seam; (d) adapt seam to have equal width throughout; (e) result after removing the seam; (f) final result after removing all seams.}
	\label{img:seamcarving}
\end{figure}

\subsubsection{Inflate-and-Push}
With \emph{Inflate-and-Push} Barth et al. \cite{Barth2014} introduced a heuristic approach to calculate layouts for word clouds using multidimensional scaling. The method uses a dissimilarity matrix $\Delta$ for which each entry is defined as
\[\delta_{i,j} = \frac{1-sim_{i,j}}{S}\] 
where $S > 0$ is some constant by which all word rectangles are scaled down.

Next, all word rectangles get iteratively \emph{inflated}, that is, the dimensions of the rectangles are increased by 5\% in each iteration. As this can lead to overlaps of words, the force-directed model described in Section \ref{sec:cpwcv} is used after every iteration to remove any overlaps. This is the ``push'' phase of the algorithm.

\subsubsection{Star Forest}\label{sec:star-forest}
This approach approximates a solution by splitting the input graph $G = (V,E)$ into disjoint stars \cite{Barth2014}.
It consists of three steps. First, $G$ has to be split into a star forest. Then, for every star, a solution has to be generated. Lastly, the generated solutions for each star have to be combined to form the final result.

The first step was realized by Barth et al. \cite{Barth2014} via a greedy algorithm. They pick the vertex $v \in V$ with the maximum sum of weights of incident edges, that is $\max_{v \in V} \sum_{u \in V} sim(u,v)$, where $sim(u,v)$ denotes the similarity or semantic relatedness of the words corresponding to vertices $u$ and $v$. The chosen vertex is then, together with a subset of its neighbors, removed from the graph. This step is repeated until there are no vertices left in $G$.

Choosing the subset of neighbors for a star center is done by reducing the problem to \textsc{Knapsack} and using the polynomial-time approximation scheme proposed in \citeapp{Lawler1979}.

The idea of the reduction is to place four rectangles $B_1, B_2, B_3, B_4$ on each corner of the rectangle $B_0$, corresponding to the central vertex, and then use the sides of $B_0$ as four bins, with the capacity of each bin being either the width or the height of $B_0$.
Four \textsc{Knapsack} instances are then solved sequentially, always removing items/vertices that have already been placed in a bin (Figure \ref{img:star-forest} (a)).

When combining the calculated solutions for the stars, semantic relationships between stars are to be preserved. For each pair of stars, a similarity value thus has to be defined. For two stars $s_1$ and $s_2$, the similarity value is simply the average similarity between the words in $s_1$ and $s_2$.
Multidimensional scaling is then used to create an initial layout from the similarity values and followed up by a force-directed algorithm to reduce white space between stars (Figure \ref{img:star-forest} (b)).
\begin{figure}[t]
	\centering
	\includegraphics[width=\textwidth]{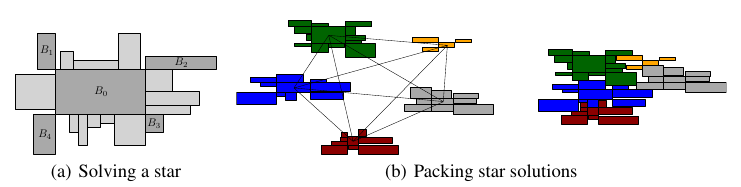}
	\caption{Realizing a star (a); Combining all stars in final result (b) \protect\cite{Barth2014}}
	\label{img:star-forest}
\end{figure}

A different approach for approximating semantic word clouds using star forests has been proposed in \cite{Barth:14}. Here, instead of \textsc{Knapsack}, the more general problem \textsc{Maximum Generalized Assignment Problem} (\textsc{Gap}) is used for the reduction.
As opposed to \textsc{Knapsack}, \textsc{Gap} allows for multiple bins with different capacity constraints, and sizes and values of items may differ depending on the bin they are put in. 
This allows for the problem to be solved in one instance, instead of one for each side of the central rectangle.
If rectangles may be rotated, a reduction to \textsc{Multiple Knapsack}, where items have the same size and value regardless of the bin they are put in, would also be sufficient.

\subsubsection{Cycle Cover}\label{sec:cycle-cover}
This approach is similar to the one described in the previous Section, in so far as the input graph, again, is split into multiple graphs with a more simple structure. That is, we want to find a vertex-disjoint cycle cover that maximizes edge weight and then realize those cycles \cite{Barth2014}.

To find the cycle cover with the maximum weight, the input graph $G = (V,E)$ is first transformed into a bipartite graph $H = (V', E')$, on which a maximum weighted matching can be more easily computed. 
To create $H$, vertices $v'$ are added for each $v \in V$  and edges $\{u',v\}, \{u,v'\}$ with weights $sim_{u,v}$ for each edge $\{u,v\} \in E$ to $G$.
From a maximum weighted matching on $H$, one can then obtain a set of vertex-disjoint paths and cycles in $G$.

To realize a cycle with $n$ vertices, its vertices are split into two paths $(v_1,...,v_{t-1})$, $(v_n,...,v_{t+1})$ and a single vertex $v_t$, such that $\sum_{i < t}w_i < \sum_{i \leq n}\frac{w_i}{2}$, where $w_i$ is the width of rectangle $v_i$.
Both paths then get aligned on a shared horizontal line. The first path $(v_1,...,v_{t-1})$ is aligned from left to right on its bottom sides, the second path $(v_n,...,v_{t+1})$ from left to right on its top sides. The leftover vertex $v_t$ is then placed with contacts to both $v_{t+1}$ and $v_{t-1}$ (Figure \ref{img:cycle-cover} (a)).

To realize a path, vertices $v_1$ and $v_2$ are initially placed next to each other. Following vertices $v_i$ are added in such a way that they are in contact with $v_{i-1}$ on its first available side, going in clockwise order and starting from the side of the contact of $v_{i-1}$ and $v_{i-2}$ (Figure \ref{img:cycle-cover} (b)).

As cycles might get very long and thus paths that create a more spiral-like, compact layout might be preferable, cycles that contain more than ten vertices are also converted into paths by removing the edge with the lowest weight.

Once all cycles and paths have been created a force-directed algorithm is used analogously to the one described in Section \ref{sec:star-forest} to reduce whitespace.

\begin{figure}[t]
	\centering
	\includegraphics[scale=1]{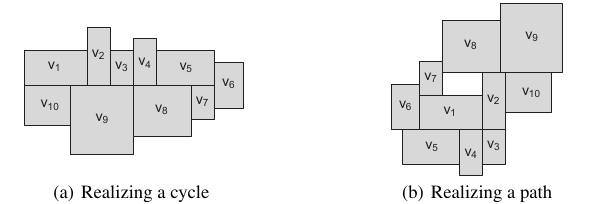}
	\caption{Realizing (a) cycles and (b) and paths.\protect\cite{Barth2014}}
	\label{img:cycle-cover}
\end{figure}

\section{Detailed construction of the split gadget.}\label{app:split_gadget}

Split gadgets are used to split variable values, so that they can be used in multiple clauses. 
We only ever duplicate variables to the right. 
Each propagation of a variable lies in a tunnel, i.e. between the two walls enclosing the corresponding variable gadget. 
Let $x_a$ be a variable such that its variable assignment ends at vertex $v_a$.
To duplicate the value of $x_a$, a second tunnel is created to the right of the existing one by adding a third wall.

The goal now is to make the variable value in the new tunnel dependent on the original one. That is, after splitting, we want two vertices $v_a'$ and $v_a''$, one in each tunnel, such that $v_a'$ realizes the same value as $v_a$ and that $v_a''$ either has the same value as $v_a$ or the value \false; see \cref{img:app_split_gadget}.
Note that this will never cause a clause gadget that is not supposed to be satisfied, to be satisfied. 

To split the values, we use a vertex $v_m$ with rectangle width $R(v_m) = 8$ that we place between layers of the inner wall.
The width of $v_m$ is set such that the rectangle can either be placed blocking half of the left tunnel or blocking half of the right tunnel, but not both.
It cannot block either tunnel completely, as this would lead to forbidden contacts; see \cref{img:split_gadget4}.

Assuming a positive value for $v_a$, where the corresponding rectangle lies on the left side in the tunnel, $v_m$ has to be placed blocking half of the tunnel $v_a$ lies in, which then allows  for $v_a''$ to be placed either on the left or on the right in its tunnel, making its value either \true or \false; see \cref{img:split_gadget1} and \ref{img:split_gadget5}.
The vertex $v_a'$ on the other hand has to stay aligned with $v_a$, as otherwise there will be fewer contacts realized.

In case of a negative value for $v_a$, $v_m$ can only be placed blocking half of the right tunnel, which in turn forces $v_a''$ to also be assigned the value \false, as there would otherwise be forbidden contacts; see \cref{img:split_gadget2} and \ref{img:split_gadget3}.
Again, $v_a'$ has to stay in line with $v_a$ to maximize contacts.

\begin{figure}[!t]
	\centering
 	\subcaptionbox{\label{img:split_gadget-contact-app}}{\includegraphics[page=2]{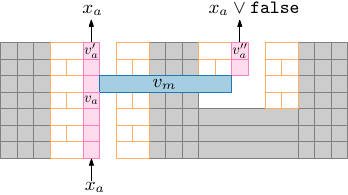}\vspace{-1.5em}}
    \hfill
	\subcaptionbox{\label{img:split_gadget1}}{\includegraphics[page=3]{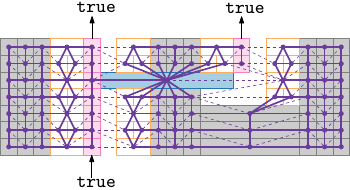}\vspace{-1.5em}}
	\subcaptionbox{\label{img:split_gadget5}}{\includegraphics[page=8]{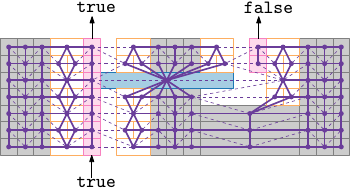}\vspace{-1.5em}}
    \hfill
	\subcaptionbox{\label{img:split_gadget2}}{\includegraphics[page=5]{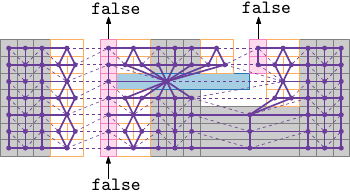}\vspace{-1.5em}}
	\subcaptionbox{\label{img:split_gadget3}}{\includegraphics[page=6]{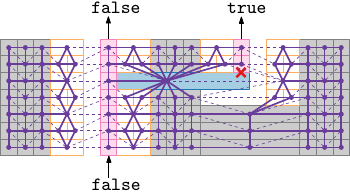}\vspace{-1.5em}}
    \hfill
	\subcaptionbox{\label{img:split_gadget4}}{\includegraphics[page=7]{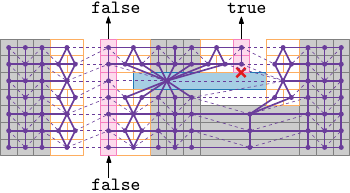}\vspace{-1.5em}}
 
	\caption{(a) Split gadget (b) and underlying graph; if $x_a$ is \true, the duplicated value is either \true or (c) \false;
	(d) valid representation of the split gadget; if $x_a$ is \false, the duplicated value has to be \false; (e) invalid representation  of the split gadget if $x_a$ is \false; (f) invalid representation of the split gadget. }
	\label{img:app_split_gadget}
	
\end{figure}

\clearpage
\section{NP-completeness of $k$-\lc}
\label{app:klc-np-complete}

\LPCrownNP*
\label{thm:lp-crown-np*}

\begin{proof}[Sketch]
    To show NP-completeness of $k$-\lc on planar graphs, we can essentially use the same gadgets we used for $k$-\ilc. That is, the vertex set and the corresponding rectangles stay the same---both in their position within a layer and in rectangle width---but we remove all edges that are not needed for the gadgets to behave as intended. Specifically, we remove edges that in the previous construction were not realized in any valid configuration of the gadgets.
    Losing the restriction to integer coordinates generally allows for rectangles to realize more contacts.
    By removing previously unused edges we keep the possible number of contacts for a rectangle equivalent to that of the construction in \cref{ch:complexity}.

    \paragraph{Variable gadget.} Consider first the variable gadget shown in \cref{img:variable_nt1}. 
    Only $v_x$ (and $u_x$ respectively) could realize $4$ ($3$) contacts not only by realizing one horizontal contact and three vertical contacts, but also by not realizing the horizontal contact but instead $4$ vertical contacts. This however would lead to fewer contacts in total; see \cref{img:variable_nt2} for an example.
    Every other vertex realizes all possible vertical contacts and thus no other placement of the rectangles would yield more contacts.

    \begin{figure}[h!]
	\centering
	\subcaptionbox{A representation with maximum number of contacts.\label{img:variable_nt1}}[.47\textwidth]{\includegraphics[page=1]{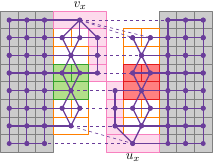}}
    \hfil
	\subcaptionbox{A shifted representation has one less contact.\label{img:variable_nt2}}[.47\textwidth]{\includegraphics[page=2]{images/nt_variable2.pdf}}
	\caption{The variable gadget in the proof of \cref{thm:lp-crown-np}}
	\label{img:nt_variable}
\end{figure}

    \paragraph{Clause gadget.}For the clause gadget, we keep all edges between the slider and the bottom layer and $v_su_a, v_su_b, v_s,u_c$,  as well as $v_sw$ where $w$ is one of the outer vertices of a wall for the upper layer; see \cref{img:nt_clause}. Similarly to the construction in \cref{ch:complexity}, this allows for the slider $v_s$ to realize $4$ contacts in case of a positive clause, and $3$ otherwise.
    
\begin{figure}[h!]
	\centering
	\includegraphics[width=\textwidth]{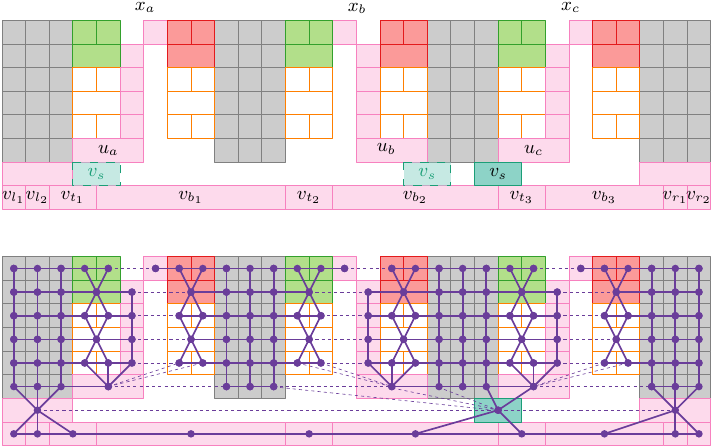}
	\caption{Contact representation (top) and underlying graph (bottom) for a clause gadget in the proof of \cref{thm:lp-crown-np}.
	Unrealized edges between $v_s$ and vertices of bottom layers as well as edges between $v_s$ and vertices of the variable gadget for $x_a$ are omitted for readability.}
	\label{img:nt_clause}
\end{figure}

    \paragraph{Split gadget.}Lastly, in the split gadget $v_m$ keeps the same edges as in \cref{img:split_gadget}, allowing for the same configurations as before; see \cref{img:nt_split}.
    Moving $v_m$ by a non-integral amount into either tunnel would decrease the number of realized contacts, as either $v_m$ loses a horizontal contact, or its left neighbor does.
\end{proof}

\begin{figure}[h!]
	\centering
	\includegraphics[page=2]{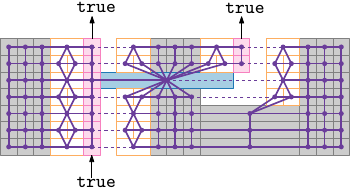}
	\caption{The split gadget in the proof of \cref{thm:lp-crown-np}.}
	\label{img:nt_split}
\end{figure} 

\clearpage
\section{Additional figures for \cref{ch:complexity}}\label{app:figures}

\begin{figure}[h]
	\centering
	\includegraphics[width=\textwidth]{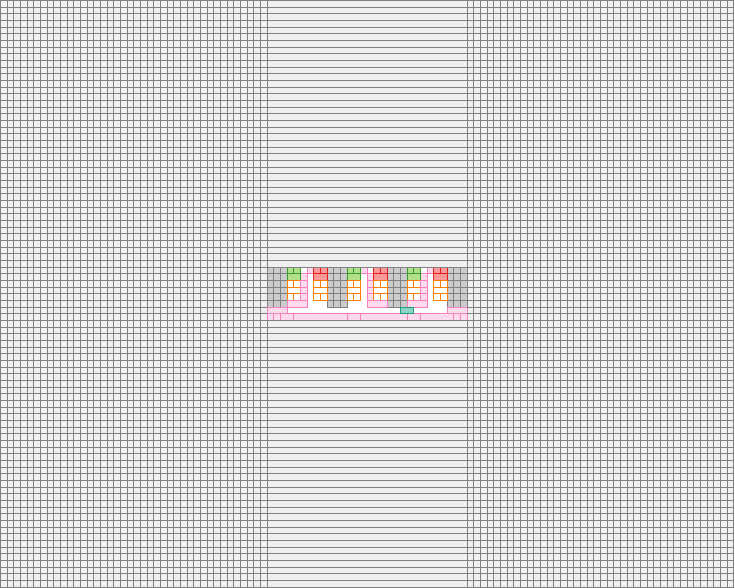}
	\caption{Frame around a formula consisting of one clause.}
	\label{img:frame}
\end{figure}

\begin{sidewaysfigure}
	\centering
	\includegraphics[width=\textwidth]{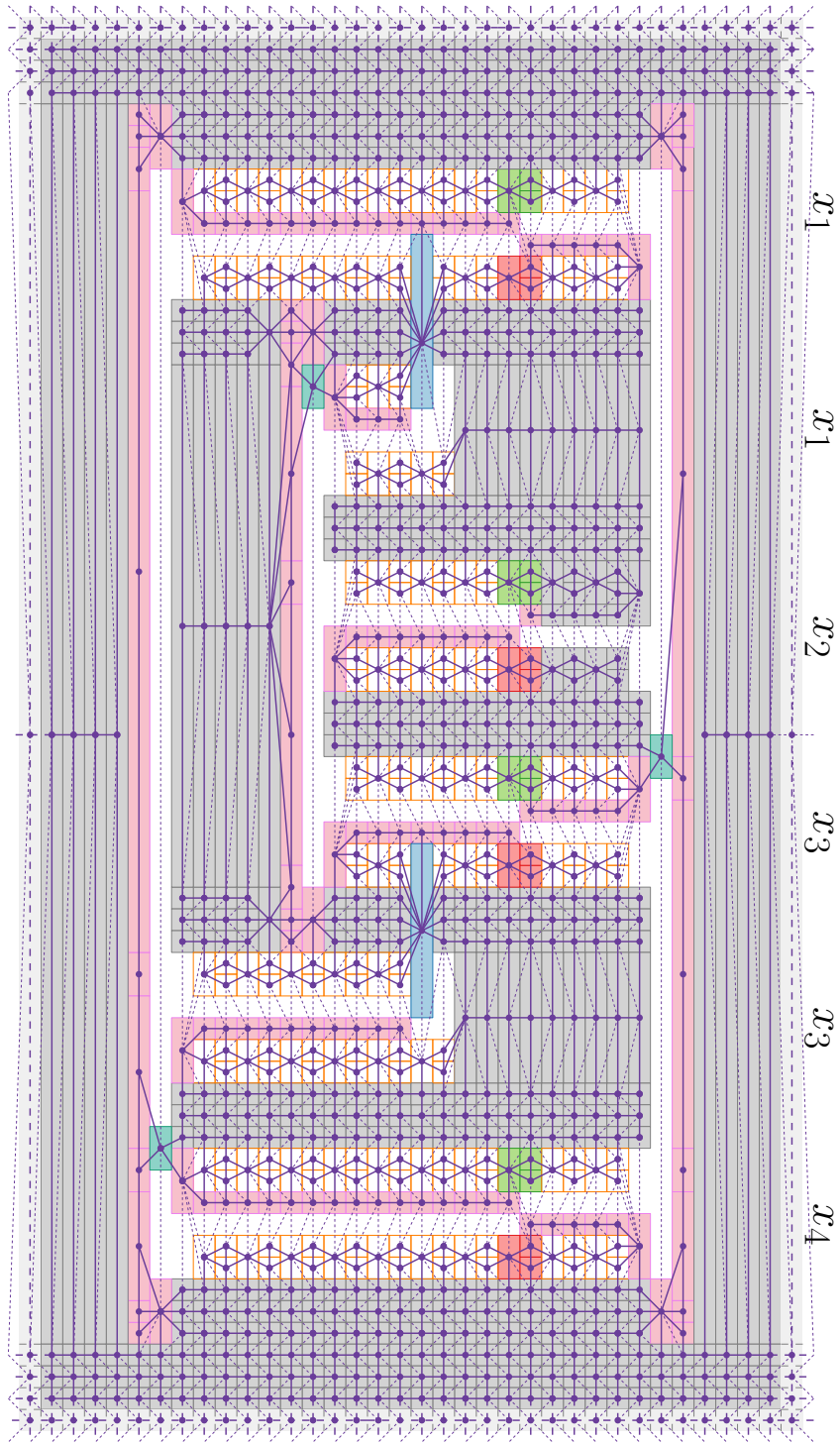}
	\caption{Contact representation for the boolean formula $B$ with variable set $\{x_1, x_2, x_3, x_4\}$, clauses $\{\overline{x_1},\overline{x_2}, \overline{x_3}\}, \{\overline{x_1}, \overline{x_3}, \overline{x_4}\}$ and $\{x_1, x_3, x_4\}$ and variable assignment $x_1 = \false, x_2 = \true, x_3 = \true, x_4 = \false$ (Edges between $v_s$ and above/below layer left out for readability purposes).}
	\label{img:example_formula_large}
\end{sidewaysfigure}

\clearpage
\bibliographystyleapp{splncs04}
\bibliographyapp{abbrv,literatur}

\end{document}